\documentclass[conference,twocolumn]{IEEEtran}%
\usepackage{amsmath}
\usepackage{amsthm}
\usepackage{pgfplots}
\usepackage{amsfonts}
\usepackage{amssymb}
\usepackage{graphicx}
\usepackage{algorithmic}
\usepackage{algorithm}
\usepackage{todonotes}
\usepackage{manfnt}
\usepackage{balance}%
\providecommand{\U}[1]{\protect\rule{.1in}{.1in}}

\theoremstyle{plain}
\newtheorem{theorem}{Theorem}

\newtheorem{lemma}{Lemma}
\newtheorem{proposition}{Proposition}
\newtheorem{problem}{Problem}
\theoremstyle{definition}

\newtheorem{definition}{Definition}

\newtheorem{condition}{Condition}
\newcommand{\C}{\mathbb{C}}

\newcommand{\abs}[1]{\vert #1 \vert}
\newcommand{\norm}[1]{\Vert #1 \Vert}

\newcommand{\sse}{\subseteq}

\usepackage{dsfont}

\newcommand{\geqsim}{\gtrsim}
\newcommand{\leqsim}{\lesssim}

\DeclareMathOperator{\argmin}{argmin}

\usepackage[nodisplayskipstretch]{setspace}
\begin{document}

\title{\scalebox{0.92}{Hierarchical Sparse Channel Estimation for Massive MIMO}}
\author{\setstretch{0.8} \scalebox{0.95}{Gerhard Wunder$^{1}$, Ingo Roth$^{2}$, Axel Flinth$^{3}$, Mahdi Barzegar$^{4}
		$, Saeid Haghighatshoar$^{4}$, Giuseppe Caire$^{4}$, Gitta
		Kutyniok$^{3}$}\\
	\scalebox{0.95}{$^{1}$Heisenberg Communications and Information Theory Group, $^{2}$Quantum Information Theory Group, FU Berlin}\\
\scalebox{0.95}{$^{3}$Applied Functional Analysis Group, $^{4}$Communications and Information Theory Group, TU Berlin}}

\maketitle

\begin{abstract}
The problem of wideband massive MIMO channel estimation is considered. Targeting for low complexity algorithms as well as small training overhead, a compressive sensing (CS) approach is pursued. Unfortunately, due to the Kronecker-type sensing (measurement) matrix corresponding to this setup, application of standard CS algorithms and analysis methodology does not apply. By recognizing that the channel possesses a special structure, termed hierarchical sparsity, we propose an efficient algorithm that explicitly takes into account this property. In addition, by extending the standard CS analysis methodology to hierarchical sparse vectors, we provide a rigorous analysis of the algorithm performance in terms of estimation error as well as number of pilot subcarriers required to achieve it. Small training overhead, in turn, means higher number of supported users in a cell and potentially improved pilot decontamination. We believe, that this is the first paper that draws a rigorous connection between the hierarchical framework and Kronecker measurements. Numerical results verify the advantage of employing the proposed approach in this setting instead of standard CS algorithms.
\end{abstract}

\section{Introduction}

Massive MIMO, i.e. deploying large number of antennas at the base station, is
a key technology for 5G \cite{5g2017_JSAC}. Although its benefits are by now
well understood and documented, the critical bottleneck for massive MIMO
deployment is still the acquisition of channel state information (CSI), with
designs attempting to balance the conflicting requirements of training
overhead and co-pilot contamination reduction, in addition to the standard
requirement of computational efficiency. This problem is even more important
in massive machine type communications (MTC) where accurate CSI with
low-training overhead is of critical importance
\cite{Wunder2015_GC,Wunder2017_ASILOMAR}.

Towards addressing these issues, one major line of works applies compressed
sensing (CS) techniques in order to account for and exploit the sparsity
properties of the wireless channel. The typical approach comes in two stages:
first the (spatial) covariance matrix of the received signal is estimated and, second, the dedicated
user channels are estimated within the estimated spatial subspaces.
This approach has been mostly applied to the narrowband signaling case,
exploiting the channel sparsity in the, so called, angle domain (see
\cite{Haghighatshoar2017_TSP} for an overview). Extensions of CS techniques to
the wideband (OFDM) massive MIMO channel have been recently considered in
\cite{ChenYang206_TWC,Haghighatshoar2017_TWC,Swindlehurst2016_TSP}. One
intuitive motivation for such an approach is that, in the wideband setting, there are more degrees of freedom available so that pilot contamination induced by co-pilots in
other cells is effectively combated.

 Initial work has been provided in
\cite{ChenYang206_TWC}, arguing that with large enough resolution in angular
and time domain, controlled by the number of antennas and subcarriers, the user
channels become approximately mutually orthogonal so that pilot contamination
diminishes. Notably, the subspace estimates are obtained through pilot
coordination over multiple slots, but not primarily through sparse channel
estimation. Moreover, the claimed orthogonality is not rigorous and crucially
depends on the parameter setting, let alone that bandwidth (spanned by the system subcarriers)
is not a free design parameter that can be set arbitrarily large. In \cite{Haghighatshoar2017_TWC}, a
computationally efficient CS-based CSI acquisition approach is proposed,
utilizing observations from only a limited number of subcarriers and antennas.
Moreover, a fine resolution of angle and time domain together with a sparse
subspace tracking algorithm is proposed. The resulting algorithm shows good
performance and is exploited for pilot decontamination purposes. A similar
setting is also considered in \cite{Swindlehurst2016_TSP}, where a minimimum
mean squared error (MMSE) channel estimator is proposed, however, requiring
prior knowledge of the channel second order statistics.

Generally, a two stage approach may require excessive observations in time in
order to obtain an accurate estimate of the received signal covariance matrix,
which may be unacceptable, e.g., in massive MTC setting. Another major problem is that
the proposed CS algorithms lack rigorous performance analysis in terms of
estimation error and, equally important, number of utilized subcarriers in
order to achieve it. This is mainly due to the specific Kronecker-like
measurement structure of the equivalent CS problem, where, as we will show,
the classical CS assumptions fail to hold, entailing convergence issues and
leakage effects. However, targeting low complexity algorithms as well as
small training overhead, say, in future MTC applications, it is imperative to
understand such structures and to derive a tailored algorithmic framework.

\textbf{Contributions}. We propose an efficient "one stage" uplink massive
MIMO wideband (OFDM) channel estimation taking into account only the sparsity
of the wideband channel into account without any additional prior knowledge.
In particular, we identify that the channel estimation problem can be posed as
the identification of a vector that is \emph{hierarchically sparse}, i.e., it
is not only sparse but its support possesses certain structural properties.
This property, together with the Kronecker-like measurement, is taken into
account for the design of an efficient CS-inspired algorithm tailored for this
particular setup. In addition, by extending the standard CS analysis
methodology to hierarchical sparse vectors, we provide a rigorous analysis of
the algorithm performance in terms of estimation error as well as number of
pilot subcarriers required to achieve it. We believe, that this is the first
paper that draws a rigorous connection between the hierarchical framework and
Kronecker measurements. Preliminary simulations show that exploitation of the
hierarchical sparsity property leads to improved estimation performance
compared to standard CS algorithms that ignore this property. Even worse,
standard algorithms may completely fail in some extreme parameter settings.
Obviously, this might have some profound impact on system parameters such as
pilot signal design, user capacity per cell etc.

\textbf{Basic notations}. $\lVert x\rVert_{\ell_{q}}:=(\sum_{i}|x_{i}%
|^{q})^{1/q}$, $q>0$, is the usual notion of $\ell_{q}$-norms, $\lVert
x\rVert:=\lVert x\rVert_{\ell_{2}}$ and $\lVert X\rVert$ is the Frobenius norm
of matrix $X$. $\mathcal{A}$ denotes a set of cardinality $|\mathcal{A}|$ and
$[N]$ denotes $\left\{  0,1,...,N-1\right\}  $. The elements of a
vector/sequence $x$ are denoted as $(x)_{i}$ (or simply $x_{i}$ if it clear
from the context). Vector $x_{\mathcal{A}}$ (matrix $X_{\mathcal{A}}$) is the
projection of elements (rows) of vector $x$ (matrix $X$) onto $\C^{\mathcal{A}%
}$. $\odot$ means point-wise (Hadamard) product, $I_{n}$ is the $n\times n$
identity matrix, $\text{diag}(x)$ is the diagonal matrix with $x\in
\mathbb{C}^{n}$ on its diagonal. $A^{H/T}$ is the Hermitian/transpose of
matrix $A$. The $N\times D$ ($N\times N$) DFT matrix is denoted as $F_{N,D}$
($F_{N,N}=F_{N}$) with $(F_{N,D})_{m,n}:=e^{-j2\pi mn/N},m\in\lbrack
N],n\in\lbrack D]$. $\mathcal{CN}\left(  0,\sigma^{2}I_{n}\right)  $ denotes
the multivariate complex Gaussian distribution of zero mean and covariance
matrix $\sigma^{2}I_{n}$. A vector $x\in\mathbb{C}^{N}$ is called $s$-sparse
if it consists of at most $s$ non-zero elements. The set of non-zero elements
(support) of $x\in\mathbb{C}^{N}$ is denoted as $\text{supp}(x)$.

\textbf{Organization of the paper:} First, we describe our signal model and
formulate the channel recovery problem. Next, we briefly present the recent
framework of hierarchical sparsity, which was developed by two of the authors
together with co-authors in \cite{Roth2016_TSP}. This framework is applied to
the channel estimation problem under the assumption of, so called, on-grid
channel parameters, and an efficient estimation algorithm is proposed. Numerical results on the performance of the algorithm are presented, demonstrating that highly accurate channel estimation can be achieved with a very small training overhead.  The
more general (and practical) case of off-grid channel parameters is then
considered where it is shown that the framework also applies with certain
modifications that take into account basis mismatch effects. 

\section{Wideband Signal Model}

We consider an uplink massive MIMO OFDM wideband channel with single-antenna
users and $M\gg1$ antenna elements at the base station corresponding to an
array manifold $a\left(  \cdot \right)  :\left[  0,\pi\right)  \rightarrow
\mathbb{C}^{M}$, which maps angular to spatial domain. Considering a uniform
linear array (ULA), the array manifold is given by
$a\left(  \phi\right)  =(1,e^{-j2\pi d\sin\phi},\ldots,e^{-j2\pi d\left(
M-1\right)  \sin\phi})^{T}$. Here, $d$ is normalized spatial separation of the
ULA, which without loss of generality (w.l.o.g.) is assumed equal to $1$ in the following. As is routinely done, we perform the change of variable $\theta
=\sin(\phi) \in [0,1)$ and, with a slight abuse of notation, we write the array manifold as a function of $2 \pi \theta$. Considering a discretized approximation of the
interval $[0,2\pi)$ by the $M$ points $\ \{k2\pi/M\}_{k=0}^{M-1}$, yields the
steering matrix $A_{\theta}:=\left[  a(0),a(1/M)\ldots,a({(M-1)/M})\right]
=F_{M}\in\mathbb{C}^{M\times M}$.

Further, suppose there are $N\gg1$ OFDM subcarriers located at the (angular)
frequencies $\omega_{0},\omega_{1},\ldots,\omega_{N-1}$, where $\omega
_{k}:=2\pi k/T_{s}$, with $T_{s}$ being the OFDM symbol duration. Assuming
that the maximum delay spread of the channel for all antennas is no greater
than a fraction $\alpha T_{s},\alpha\leq1$, which is the case in any
reasonable OFDM design, the \textquotedblleft delay manifold\textquotedblright%
\ $b\left(  \cdot \right)  :\left[  0,\alpha T_{s}\right]  \rightarrow
\mathbb{C}^{N}$ is defined as $b\left(  \tau\right)  :=(e^{-j\omega_{0}\tau
}\;e^{-j\omega_{1}\tau}...e^{-j\omega_{N-1}\tau})^{T}$, which maps the delay
to the frequency domain. Considering a discretized approximation of
$[0,T_{s}]$ by the $N$ points $\{kT_{s}/N\}_{k=0}^{N-1}$, yields the steering
matrix $A_{\tau}:=\left[  b(0),b(T_{s}/N),\ldots,b((D-1)T_{s}/N)\right]
=F_{N,D}\in\mathbb{C}^{N\times D}$ where sample number $D\leq N$ is the
discrete delay spread\footnote{In general, a denser discretized approximation
for the angle and delay domains could be employed. We leave investigation of
this case for future work.}.

The channel of any user is a superposition of a small number $L$ of impinging wavefronts
(paths) characterized by their delay/angle pairs $\{(\tau_{p},\theta
_{p})\}_{p=0}^{L-1}$, with $\tau_{p}\in\lbrack0,\alpha T_{s}]$, $\theta_{p}%
\in\lbrack0,1)$, whose values are assumed to remain constant during the
considered transmission interval. The channel \textquotedblleft
spatial-frequency\textquotedblright\ transfer matrix $H\left(  t\right)
\in\mathbb{C}^{N\times M}$ of an arbitrary user corresponding the OFDM symbol
(slot) index $t\in\mathbb{Z}$ can be then be written as \cite{ChenYang206_TWC}%
\begin{equation} {\label{eq:channel matrix}}
H\left(  t\right)  =\sum_{p=0}^{L-1}\rho_{p}\left(  t\right)  b\left(
\tau_{p}\right)  a^{H}\left(  \theta_{p}\right)  ,
\end{equation}
where $\rho_{p}(t)\in\mathbb{C}$ is the time-varying gain of the $p$-th path
at slot $t$.

Assuming that the base station observes $T$ consecutive (pilot) OFDM slots dedicated for channel estimation, the
overall received signal equals
\begin{equation} 
Y\left(  t\right)  =\text{diag}(c(t))H(t)+Z\left(  t\right)  ,\;t\in\lbrack
T].\label{eq:observed_signal}%
\end{equation}
The matrix $Z\left(  t\right)  \in\mathbb{C}^{N\times M}$ represents noise
with independent, identically distributed elements as $\mathcal{CN}%
(0,\sigma^{2})$. Vector $c(t):=\left(  c\left(  t,\omega_{0}\right)
,...,c(t,\omega_{N-1})\right)  \in\mathbb{C}^{N}$ contains the pilot symbols
transmitted over slot $t$. For the single-cell case considered in this paper,
it is reasonable to assume that users are assigned non-overlapping sets of
subcarriers to transmit their pilots on, which, w.l.o.g., are set to unity, i.e., $\text{diag}(c(t))=I_{N}$.
Leveraging the sparse nature of the wideband massive MIMO channel, estimation
of the channel of an arbitrary user can be achieved in principle by considering only
the observations from the $O_{\tau}\leq N$ pilot subcarriers the user in consideration utilizes
and $O_{\theta}\leq M$ antennas. Let $P_{\tau}\in\{0,1\}^{O_{\tau}\times N}$
and $P_{\theta}\in\{0,1\}^{O_{\theta}\times M}$ denote the corresponding
\emph{sampling} matrices in frequency and space dimensions, respectively. We consider a random sampling in frequency and space,
i.e., the $O_\tau$ pilot subcarriers are selected uniformly from the $N$ available subcarriers and similarly for the sampled antennas.


In principle, identification of the continuous-valued channel parameters
$\{(\rho_{p},\tau_{p},\theta_{p})\}_{p=0}^{L-1}$ can be obtained from the
low-dimensional sketches $\{P_{\tau}Y(t)P_{\theta}^{T}\}_{t=0}^{T-1}$ using
the, so called, (two-dimensional) super-resolution approach (see, e.g.,
\cite{heckel2014SuperResRadar}), which is an extension of the one dimensional
super-resolution approaches (see, e.g., \cite{tan2014direction},
\cite{BarzegarEtAl:2017}). Unfortunately, the numerical resolution of the
two-dimensional problem is computationally intensive, rendering this approach
practical only for scenarios with $O_\tau$ and $O_\theta$ up to about $10$ each.

Targeting low-complexity channel estimation, we consider a discretized
representation of the channel matrix, which translates the physical channel sparsity to
sparsity of an appropriately defined matrix that is to be identified by the
estimator. In particular, let us first assume that each delay/angle pair lies
exactly on the delay/angle grid corresponding to the steering matrices
$A_{\theta}$ and $A_{\tau}$, i.e., it holds $(\tau_{p},\theta_{p})=(k_{p}%
T_{s}/N,l_{p}2\pi/M)$ for some $k_{p}\in\lbrack N]$ and $l_{p}\in\lbrack M]$,
for all $p\in\lbrack L]$. It is easy to see that, in this case, $H(t)$ can be
written as
\begin{equation} 
H(t)=A_{\tau}W\left(  t\right)  A_{\theta}^{H}%
,\label{eq:channel_delay_angle_decomposition}%
\end{equation}
where
\begin{equation} {\label{eq: W on grid}}
W\left(  t\right)  :=\sum_{p=0}^{L-1}\rho_{p}(t)e_{k_{p}}e_{l_{p}}^{T}%
\in\mathbb{C}^{D\times M},
\end{equation}
with $e_{n}$ denoting the canonical basis vector of appropriate dimension with
the $n$-th element equal to $1$. Matrix $W(t)$ is the \emph{delay-angular
representation} of the channel at time $t$, which is a sparse matrix, with
only $L$ nonzero elements out of a total $MD$. Note that the set of non-zero
elements (support) of $W(t)$ is the same for each $t$, although $W(t_{1})\neq
W(t_{2})$, for $t_{1}\neq t_{2}$ due to the time-varying path gains.

The general case, i.e. when the delay/angle pairs are off the delay/angle
grid, is a bit more subtle as the discretized model leads to \emph{basis
mismatch} effects \cite{chi2011sensitivity}. In particular, even though the
representation of (\ref{eq:channel_delay_angle_decomposition}) is still valid,
the delay-angular representation matrix now equals
\begin{equation}
W(t)=\sum_{p=0}^{L-1}\rho_{p}(t)u_{\tau_{p}}u_{\theta_{p}}^{H}\in
\mathbb{C}^{D\times M},\label{eq:W_basis_mismatch}%
\end{equation}
where $u_{\tau_{p}}\in\mathbb{C}^{D\times1},u_{\theta_{p}}\in\mathbb{C}%
^{M\times1}$ are defined by the equations $b(\tau_{p})=A_{\tau}u_{\tau_{p}}$,
$a(\theta_{p})=A_{\theta}u_{\theta_{p}}$, respectively, for all $p\in\lbrack
L]$. In general, $W(t)$ no longer consists of only $L$ non-zero elements due
to energy leakage over the grid points \cite{chi2011sensitivity} (see Fig.
\ref{fig grid and off-grid W}). As we will see, it can however (under a
regularity condition) be approximated by a matrix with a so-called
\emph{hierarchical sparsity pattern}, for which recovery methods have recently
been developed in \cite{Roth2016_TSP}.

\begin{figure}[ptb]
\centering
\includegraphics[width=0.6\columnwidth]{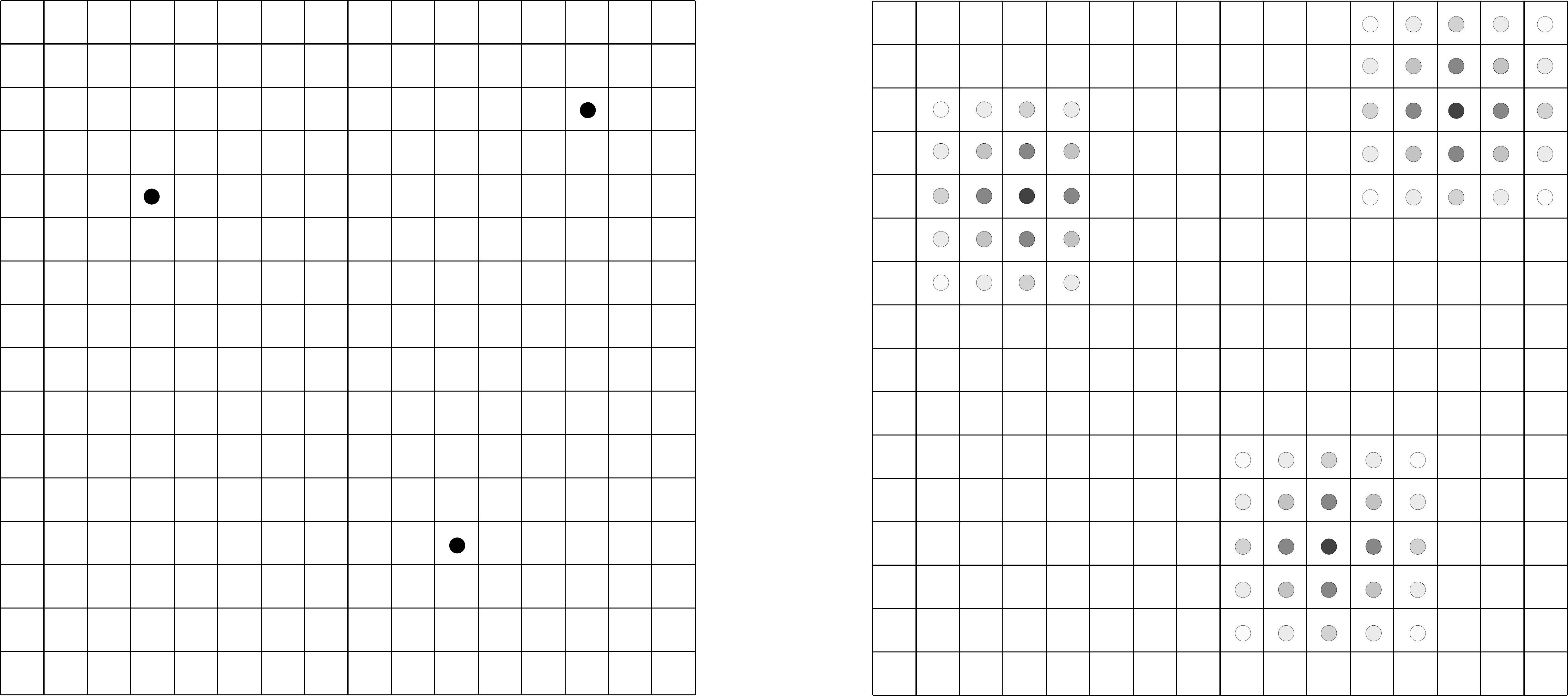}\caption{Example of
support of the delay-angular channel representation $W(t)$ for $L=3$ impinging
wavefronts. Left: on grid case, Right: off-grid case.}%
\label{fig grid and off-grid W}%
\end{figure}

\section{Problem Statement}

Considering the estimation of the channel of a single user, the problem can be
formulated as the estimation of the set of matrices $\{W(t)\}_{t=0}^{T-1}$
given the low-dimensional sketches $\{P_{\tau}Y(t)P_{\theta}^{T}\}_{t=0}%
^{T-1}$. For technical reasons, the observed matrices $\{Y(t)\}_{t=0}^{T-1}$ are normalized by
$1/\sqrt{O_\tau O_\theta}$, which, of course, occurs no information loss, and, for convenience, are vectorized,
resulting in
\begin{align}
X\left(  t\right)   &  :=\frac{1}{\sqrt{O_\tau O_\theta}}\text{vec}\left(  P_{\tau}Y(t)P_{\theta}^{T}\right)
\nonumber\\
&  =\Psi\text{vec}\left(  W\left(  t\right)  \right)  +\Phi\text{vec}\left(
Z\left(  t\right)  \right)  ,t\in\lbrack T],
\end{align}
where $\Phi:=(1/\sqrt{O_\tau O_\theta})P_{\theta}\otimes P_{\tau}\in\{0,1\}^{O\times MN}$ with
$O:=O_{\tau}O_{\theta}$ and
\begin{equation}
\Psi:=\left(  \sqrt{\frac{1}{O_{\theta}}}P_{\theta}A_{\theta}^{\ast}\right)
\otimes\left(  \sqrt{\frac{1}{O_{\tau}}}P_{\tau}A_{\tau}\right)  \in
\mathbb{C}^{O\times MD} \label{eq: Psi_definition}%
\end{equation}
is the, so called, \emph{sensing matrix} \cite{MathIntroToCS} for the
observations (measurements). By repeating this vectorization procedure over
the slot dimension as well, the problem can be formulated as follows.

\begin{problem}
Find a computationally efficient estimator of
\[
\bar{W}:=\emph{\text{vec}}\left(  \left[  \emph{\text{vec}}(W(0)),\ldots
,\emph{\text{vec}}(W(T-1))\right]  ^{T}\right)  \in\mathbb{C}^{MDT}%
\]
given the vector consisting of multiple measurements%
\[
\bar{X}:=\emph{\text{vec}}\left(  \left[  \emph{\text{vec}}(X(0)),\ldots
,\emph{\text{vec}}(X(T-1))\right]  ^{T}\right)  \in\mathbb{C}^{OT},
\]
under the linear model
\[
\bar{X}=\bar{\Psi}\bar{W}+\bar{Z},
\]
where $\bar{Z}:=\emph{\text{vec}}\left(  \left[ \Phi \emph{\text{vec}%
}(Z(0)),\ldots,\Phi \emph{\text{vec}}(Z(T-1))\right]  ^{T}\right)  $, $\bar{\Psi
}:=\Psi\otimes I_{T}\in\mathbb{C}^{OT\times MNT}$ with $\Psi$ as given in
(\ref{eq: Psi_definition}) and with the sampling matrices $P_{\tau}%
\in\{0,1\}^{O_{\tau}\times N}$ and $P_{\theta}\in\{0,1\}^{O_{\theta}\times M}$
generated by randomly and independently selecting $O_{\tau}$ and $O_{\theta}$
rows from the identity matrices $I_{N}$ and $I_{M}$, respectively. We also ask
for the scaling of required antennas $O_{\theta}$ and subcarriers $O_{\tau}$
for reliable (in a specific approximate sense that we will made precise later) recovery.
\end{problem}

Note that we have assumed random subsampling matrices in frequency and antenna
spaces, which, in turn, imply random sets of pilot subcarriers. Even though
not necessarily optimal, this random subsampling approach simplifies analysis
and is actually shown to achieve good performance.

A naive application of standard results from compressive sensing theory
\cite{MathIntroToCS} suggests that recovery is guaranteed as soon as
\begin{equation}
TO\gtrsim LT\log({T N M}).
\end{equation}
However, this result only holds as long as the sensing matrix of the problem
satisfies certain properties, e.g. the restricted isometry property (RIP)
\cite{MathIntroToCS}. As we will see, the Kronecker-like sensing matrix of
this problem \emph{might not} possess such properties in general, rendering a
different algorithm design and analysis methodology necessary. To this end,
the appropriate mathematical framework is developed in the next section by
exploiting the specific sparsity structure of $\bar{W}$.
\section{Algorithm Design Exploiting Structural Properties of Sparsity}

This section identifies important structural properties of the channel
matrices $\{W(t)\}_{t=0}^{T-1}$, which are taken into account for designing
efficient and accurate algorithms, as well as obtaining rigorous performance
guarantees. For this section, the case where the path/delay values of each
path lie exactly on the sampling grid is considered, i.e., $W(t),t\in[T]$,
consists of exactly $L$ non-zero elements. The general case will be treated in
the next section.

\subsection{Hierarchical Sparsity and Algorithm Design}

The fundamental observation towards an efficient channel estimation algorithm
is that the delay/angle channel representation is not only sparse, but
possesses certain structural properties as well, that fall within the notion of \emph{hierarhical sparsity}.

\begin{definition}
[Hierarchical sparsity] Let $\mathbf{s} = (s_{1}, \dots, s_{\ell})$ be $\ell
$-tuples of natural numbers and consider a vector $\tilde{X} \in
\mathbb{C}^{N_{1} N_{2} \cdots N_{l}}$, with integer $N_{i}\geq s_{i}%
,i\in\{1,2,\ldots, l\}$. We define the hierarchical $\mathbf{s}$-sparsity of
$\tilde{X}$ inductively as follows: For $l=1$ it is exactly the same as the
standard notion of sparsity (i.e., at most $s_{1}$ out of $N_{1}$ elements are
non-zero). For $l>1$, $\tilde{X}$ is called $\mathbf{s}$-sparse if it consists
of $N_{1}$ blocks out of which at most $s_{1}$ are non-zero and each of the
non-zero blocks is $(s_{2}, \dots, s_{\ell})$-sparse.
\end{definition}

\begin{figure}[t]
\centering
\includegraphics[width=6.5cm]{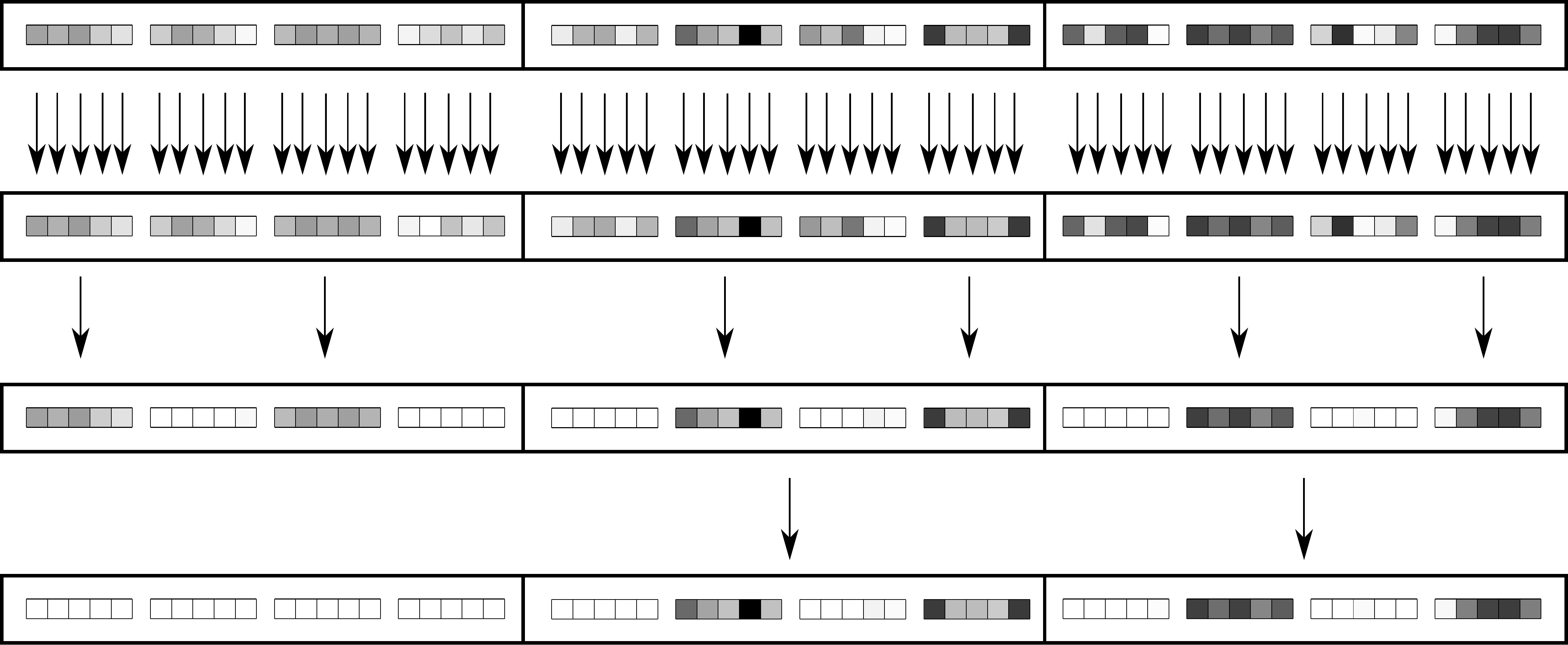}\caption{Calculation of the
action of the $\mathcal{T}_{2,2,5}$-operator on a vector $X \in\C^{ 3 \cdot4
\cdot5}$}%
\label{fig:Tappl}%
\end{figure}

It is easy to see that the unknown vector $\bar{W}\in\mathbb{C}^{MDT}$, defined
in Sec. III, is a level $l=3$ compound vector that is $(L,L,T)$-sparse.
Actually, for sufficiently large $M$, one may also expect that no more that
$K$ ($1\leq K\leq L$) paths have the same angle, implying further restricting
$\bar{W}$ as an $(L,K,T)$-sparse vector.

Clearly, the hierarchically sparse structure  $\bar{W}$ \emph{should be exploited} in algorithm design and analysis as it provides
significant restrictions on the support of $\bar{W}$, compared to the standard
notion of sparsity (which would characterize $\bar{W}$ simply as $TL$-sparse).
Towards this end, the low-complexity, hierarchical hard thresholding pursuit
(HiHTP) algorithm is considered here, originally proposed in
\cite{Roth2016_TSP} for the estimation of general hierarchically sparse
vectors from linear measurements. For the channel estimation problem, it takes
the following form.

\begin{algorithm}
[h] \caption{HiHTP} \label{alg:HiHTP}

\begin{algorithmic}
[1]

\REQUIRE measurement $\bar{X}$, sensing matrix $\bar{\Psi}$, $\left(
L,K,T\right)  $ hierarchical sparsity for $\bar{W}$

\STATE$\hat{\bar{W}}^{\left(  i\right)  }=0$

\REPEAT

\STATE$\mathcal{A}^{\left(  i+1\right)  }=\mathcal{T}_{(L,K,T)}\left(
\hat{\bar{W}}^{\left(  i\right)  }+\bar{\Psi}^{H}\left(  \bar{X}-\bar{\Psi}
\hat{\bar{W}}^{\left(  i\right)  }\right)  \right)  $

\STATE$\hat{\bar{W}}^{\left(  i+1\right)  }=\arg\min{}_{\substack{M\in
\mathbb{C}^{MNT} , \\\text{supp}(M)\subseteq\mathcal{A}^{\left(  i+1\right)
}}}\left\{  \Vert\bar{X}-\bar{\Psi} M\Vert\right\}  $

\UNTIL stopping criterion is met at $i=i^{\ast}$

\ENSURE$(L,K,T)$-sparse matrix $\hat{\bar{W}}^{\left(  i^{\ast}\right)  }$
\end{algorithmic}
\end{algorithm}

The algorithm follows the philosophy of model-based compressed
sensing~\cite{BaraniukEtAl:2010}: In each iteration, it first makes a gradient
descent step towards minimizing a least square objective, then projects the
resulting signal onto the $(L,K,T)$-sparse support containing the most of its
energy via application of operator $\mathcal{T}_{(L,K,T)}(\cdot)$, and
subsequently solves a least squares problem restricted to that support to find
the next iterate.

Utilization of the projection (or thresholding) operator $\mathcal{T}%
_{(L,K,T)}$ is the main differentiator of the HiHTP algorithm compared to the
standard HTP algorithm \cite{MathIntroToCS}. In particular, for any compound
vector $\tilde{X}\in\mathbb{C}^{N_{1}N_{2}\cdots N_{l}}$ and any
$\mathbf{s}:=(s_{1},s_{2},\ldots,s_{l})$, $\mathcal{T}_{\mathbf{s}}(\tilde
{X})$ is defined as
\begin{equation}
\mathcal{T}_{\mathbf{s}}(\tilde{X}):=\operatorname{supp}\operatorname*{arg}%
_{\mathbf{s}\text{-sparse }{\tilde{Z}}}\ \Vert\tilde{X}-\tilde{Z}\Vert.
\end{equation}
The action of this operator can be computed very efficiently. First, at the
lowest level, it selects the $s_{l}$ largest-magnitude entries (out of a total
$N_{l}$ entries) for each sub-block. Then, iteratively, at level $k<l$, it
selects the $s_{k}$ sub-blocks (out of a total $N_{k}$ sub-blocks) whose best
$(s_{k+1},\ldots,s_{l})$-sparse approximation are largest in $l_{2}$-norm. As
an example, the iterative calculation of $\mathcal{T}_{2,2,5}$ applied on a
vector $X\in\mathbb{C}^{3\cdot 4 \cdot 5}$ is illustrated in Fig. \ref{fig:Tappl}.

\subsection{Performance analysis}

Towards characterizing the algorithm performance, the authors of
\cite{Roth2016_TSP} introduced the concept of Hierarchical RIP (HiRIP) constant.

\begin{definition}
[HiRIP]\label{def:RIP} Given a matrix $\tilde{\Psi}\in\mathbb{C}^{O\times
N_{1} \cdots N_{\ell}}$ and a vector $\mathbf{s}=(s_{1},s_{2},\ldots,s_{l})$,
we denote by $\delta_{\mathbf{s}}$ the smallest $\delta\geq0$ such that
\begin{equation}
\label{eq:RIPequation}(1-\delta)\Vert\tilde{X}\Vert^{2}\leq\Vert\Psi\tilde
{X}\Vert^{2}\leq(1+\delta)\Vert\tilde{X}\Vert^{2},
\end{equation}
for all $\mathbf{s}$-sparse vectors $\tilde{X}\in C^{N_{1} \cdots N_{\ell}}$.
$\tilde{\Psi}$ is said to have the hierarchical RIP (HiRIP).
\end{definition}

Note the definition of HiRIP is less general than standard RIP: Since
$\mathbf{s}$-sparse vectors in particular are $s_{1} \cdots s_{\ell}$-sparse,
$s_{1} \cdots s_{\ell}$-RIP implies $\mathbf{s}$-HiRIP, whereas $\mathbf{s}%
$-HiRIP does not necessary imply $s_{1} \cdots s_{\ell}$-RIP. However, this
more restricted notion of RIP is well justified here as it explicitly takes
into account the hierarchical sparsity property of the vectors that are
considered in our problem.

Using the HiRIP concept, the following recovery guarantee for the HiHTP
algorithm can be stated.

\begin{theorem}
\label{thm:HiHTP} (Recovery guarantee \cite{Roth2016_TSP}) Let $3\times
(L,K,T):=(\min(3L,M),\min(3K,D),T)$ and suppose that the sensing matrix
$\bar{\Psi}$ in Algorithm \ref{alg:HiHTP} has a HiRIP constant
\begin{equation}
\delta_{3\times(L,K,T)}<\frac{1}{\sqrt{3}}. \label{eq:recGarant:RIP}%
\end{equation}
Then, the sequence of estimates $\{\hat{\bar{W}}^{\left(  i\right)  }\}$ of
the HiHTP algorithm satisfies
\begin{equation}
\Vert\bar{W}-\hat{\bar{W}}^{\left(  i+1\right)  }\Vert\leq\kappa^{i}\Vert
\bar{W}-\hat{\bar{W}}^{\left(  0\right)  }\Vert+\tau\Vert\bar{Z}\Vert,
\end{equation}
for any $i\geq0$, with
\begin{equation}
\kappa:=\left(  {\frac{2\delta_{3\times(L,K,T)}}{1-\delta_{3\times(L,K,T)}%
^{2}}}\right)  ^{1/2}<1
\end{equation}
and $\tau\leq5.15/(1-\kappa)$.
\end{theorem}

It follows that, in order to ensure accurate recovery of $\bar{W}$ via the
HiHTP algorithm and assuming $M>3L$, $N>3K$, we need to estimate the
$(3L,3K,T)$-HiRIP constants for $\bar{\Psi}$. %
%
%
We will use the following bound for general Kronecker-type sensing
matrices \cite{RothEtAl2017}.

\begin{theorem}
\label{thm:HiRIP} Suppose $\tilde{\Psi}:=M_{1} \otimes M_{2} \otimes
\dots\otimes M_{\ell}$, where $M_{k}\in
\mathbb{C}^{O_{k}\times N_{k}}$ for all $k$. Further suppose that, for each $k$, $M_{k}$ has $s_{k}$-sparse RIP with constant
$\delta_{s_{k}}$. Then, with $\mathbf{s}=(s_1,\ldots,s_k)$, $\tilde{\Psi}$ satisfies (\ref{eq:RIPequation}) with a HiRIP
constant
\begin{equation}
\delta_{\mathbf{s} }\leq\prod_{k=1}^{l} (1+ \delta_{s_{k}})-1.
\end{equation}

\end{theorem}

The above theorem implies that sensing matrices resulting by Kronecker
products will have HiRIP, provided each of the constituent matrices has the
(standard) RIP.%
Notably, it is important to emphasize that while HiRIP is attainable, RIP may
actually not: For example, it can be shown \cite{JokarMehrmann} that the
(standard RIP-properties of a Kronecker product $M_{1}\otimes\dots\otimes
M_{\ell}$ cannot be better than the $RIP$-properties of the weakest matrix (with respect to the total sparsity!)
$M_{i}$, since $\delta_{\sum_k s_k}(M_{1}\otimes\dots\otimes M_{\ell})\geqsim\max
_{i=1}^{\ell}\delta_{\sum_k s_k}(M_{i})$. Hence, to consider the HiRIP / HiHTP
framework, instead of the standard RIP framework, is \emph{inevitable} to
obtain recovery guarantees when using sensing matrices we consider in this publication.

\subsection{Numerical Results}
This section demonstrates the effectiveness of the HiHTP in obtaining highly accurate channel estimates with limited pilot overhead. In all cases, an OFDM system with $N=64$ subcarriers and $M=16$ antennas at the base station is considered. The ``spatial-frequency'' transfer matrix of the user in consideration is represented as in (\ref{eq:channel_delay_angle_decomposition}) and (\ref{eq: W on grid}), i.e., with ``on grid'' angle/delay values for each path, with $L=3$. The angle/delay values of each path remain constant for the $T$ slots considered in the estimation procedure, whereas the channel gains are independent and identically distributed (i.i.d.) over slots. For each slot, the channel gains are generated as i.i.d. complex Gaussian variables with a total power $\sum_{p=0}^{L-1}.  \mathbb{E}(|\rho_p(t)|^2)=1$. The path angles $\{\theta_p\}_{p=0}^{L-1}$ are generated independently and uniformly over the angle sampling grid, however, no two paths are allowed to have the same angle. The path delays are independent and uniformly distributed over the delay sampling grid with $D=16$. For this moderate number of antennas it is reasonable to consider all of them for channel estimation purposes, i.e., $O_\theta=M$, however, the number of (randomly) selected pilot subcarriers $O_\theta$ is a design variable.

Figure \ref{fig:MSE_vs_O_tau} depicts the (average) mean squared error (MSE), $\frac{1}{NM}\mathbb{E}\|H(t) - \hat{H}(t) \|^2$, of the space-frequency transfer matrix estimate obtained as $\hat{H}(t):= A_\tau \hat{W}(t) A_\theta^H$ where $\hat{W}(t)$ is the estimate of the delay-angular channel representation at slot $t$ provided by HiHTP. The MSE is depicted as a function of the training overhead $O_\tau/N$ and for various values of observed slots $T$. The signal-to-noise ratio (SNR) was set equal to $1/\sigma^2=0$ dB. It can be seen that HiHTP offers excellent estimation accuracy for a very limited training overhead. For example, for $T=1$, a training overhead of around $0.08$ is sufficient to achieve a MSE that is almost one order of magnitude less than the noise level. This overhead should be compared with conventional estimation approaches which would require a pilot overhead in the order of $D/N=0.25$. As expected, jointly considering $T>1$ slots improves performance as the algorithm incorporates the common support of the delay-angular representation of the channel over slots. The gain of increasing $T$ is more visible in the low training overhead regime, with $T=4$ sufficient to achieve almost all of the possible gain. As a comparison, the performance of the standard HTP algorithm \cite{MathIntroToCS} is depicted. It can be seen that, for small training overhead, HTP performance is significant worse than HiHTP, exactly due to not taking into account the hierarchical sparsity of the channel. For sufficiently large training overhead, the performance of both algorithms are the same, implying that knowledge of the sparsity structure plays no role, exactly analogous to standard estimation theory where a priori information becomes irrelevant once sufficiently many observations are obtained.
\begin{figure}
	\centering
	\includegraphics[width=\columnwidth]{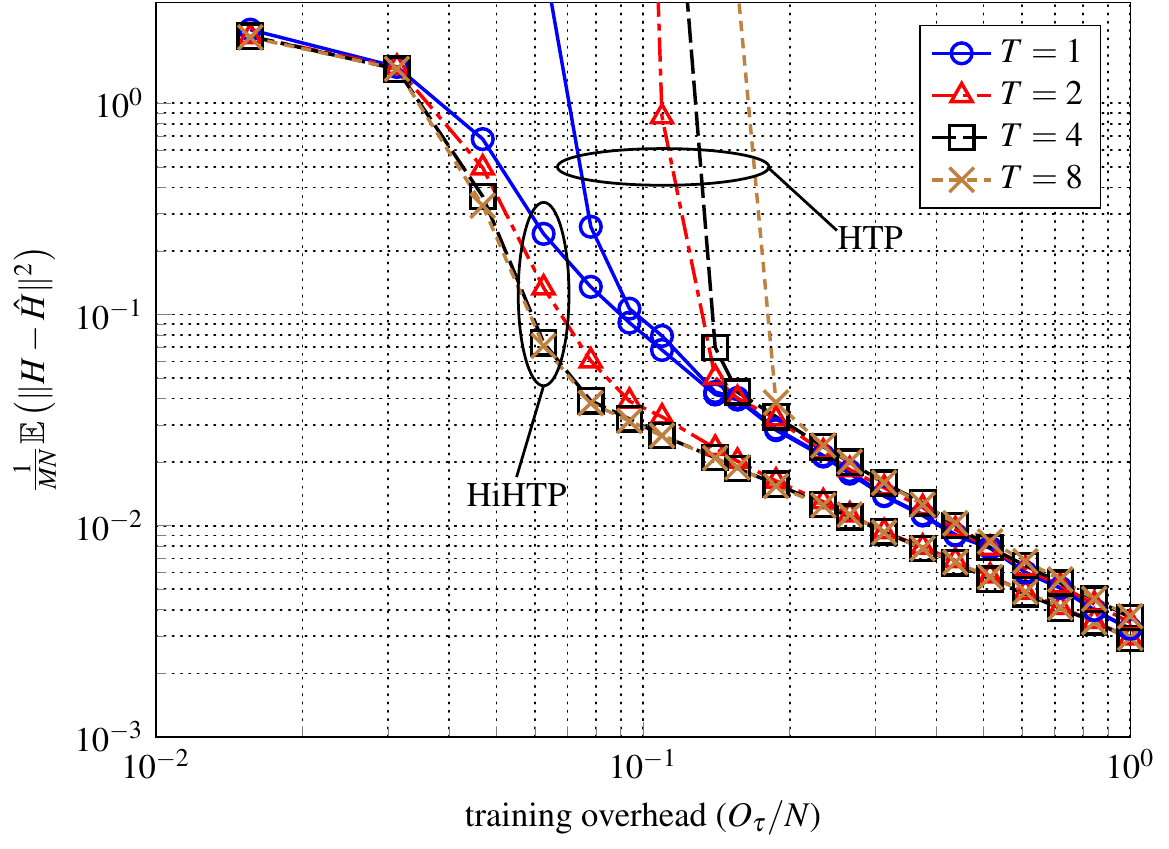}\caption{MSE of HiHTP/HTP as a function of training overhead, for various values of $T$ ($N=64, M=16, D=16, L=3, \mathrm{SNR} = 10 \text{ dB}, O_\theta=M$).}%
	\label{fig:MSE_vs_O_tau}%
\end{figure}

Figure \ref{fig:MSE_vs_SNR} shows the performance of HiHTP for a training overhead $O_\tau/N$ = 0.125 as a function of SNR and for various $T$. As expected MSE performance improves with SNR. Utilizing more than one slots is beneficial for further improving performance in the low SNR regime.
\begin{figure}
	\centering
	\includegraphics[width=\columnwidth]{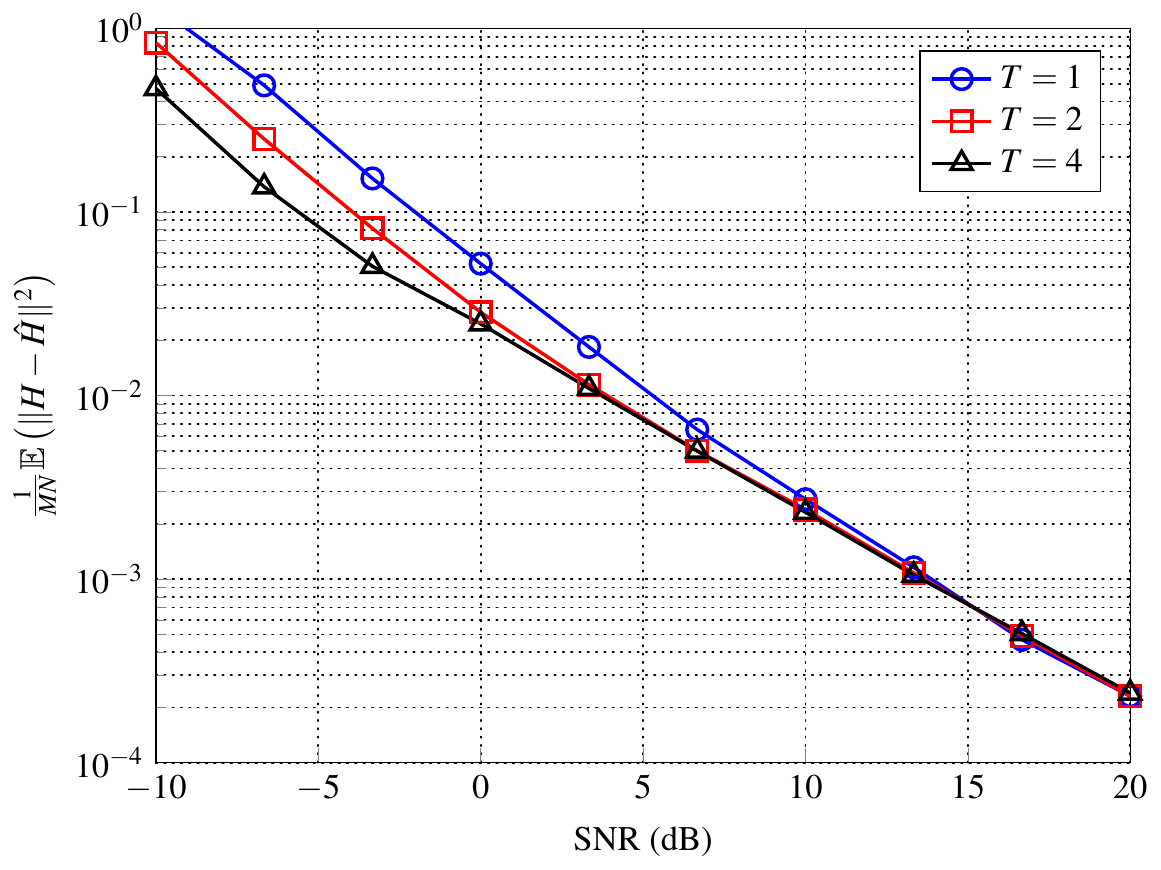}\caption{MSE of HiHTP as a function of SNR ($1/\sigma^2$) for various values of $T$ ($N=64, M=16, D=16, L=3, O_\tau=0.125 N, O_\theta = M$).}%
	\label{fig:MSE_vs_SNR}%
\end{figure}

\section{Hierarchical Sparsity and HiHTP Performance Under Basis Mismatch}

The previous section identified the structural properties of the delay/angular
channel representation which were exploited in order to obtain an efficient
estimation algorithm (HiHTP) and obtain performance guarantees using the
concept of HiRIP. However, the analysis considered the on-grid case, i.e.,
with the delay/angle values lying on the grid assumed by the steering matrices
$A_{\tau}$ and $A_{\theta}$. As mentioned in Sec. II, this is not the case in
general leading to basis mismatch effects. This section considers this case
with the analysis split into two parts: First, we show that the matrix
$\bar{W}$ can be approximated by a hierarchically sparse matrix even under
basis mismatch. Then, we provide conditions in terms of number of subcarriers
($O_{\tau}$) and number of antennas ($O_{\theta}$) that should be considered
in order to achieve (approximate) recovery of $W$ for a single user.

\subsection{Sparse Approximation Analysis}

As has been argued in the introduction, the matrix $W(t),t\in[T]$ is not
exactly sparse in general. In order to get a concrete recovery result for the
application of the HiHTP algorithm for this case as well, we need to quantify
how well it can be approximated as (hierarchically) sparse. In order to do
this, we first analyze the sparse approximation error for the vectors
$u_{\tau_{p}}$ and $u_{\theta_{p}}, p \in[L]$ appearing in
(\ref{eq:W_basis_mismatch}). For simplicity, the case $T=1$ will be considered
with the results trivially extending to the multiple measurements case.

The vector $u_{\theta}$, for $\theta$ arbitrary, is easily seen to be equal to
$(1/{M})F_{M}^{H} a(\theta)$. As shown in \cite{ChenYang206_TWC}, the
$i$-th element of $u_{\theta}$ equals
\begin{align}
(u_{\theta})_{i}  &  =\frac{\sin\left(  \pi M\xi_{i}\left(  \theta\right)
\right)  }{M\sin\left(  \pi\xi_{i}\left(  \theta\right)  \right)  }%
e^{-j\pi\left(  M-1\right)  \xi_{i}\left(  \theta\right)  }, i \in[M],
\end{align}
where $\xi_{i}\left(  \theta\right)  := \theta- i/M$ is a translation of the
angle $\theta$. In the exact same fashion and for any value of $\tau$, it
holds
\begin{align}
(u_{\tau})_{i}  &  =\frac{\sin\left(  \pi N\eta_{i}(\tau)\right)  }%
{N\sin\left(  \pi\eta_{i}\right)  }e^{-j\pi\left(  N-1\right)  \eta_{i}(\tau
)},i \in[N],
\end{align}
where $\eta_{i}(\tau):=\tau-i/N$. The sparse approximation properties of
$u_{\tau}$ and $u_{\theta}$ are identified in the following lemma.

\begin{lemma}
\label{lem:SparseAppr} For any value of $\theta$, there exists a $(2K_{\theta
}+1)$-sparse vector $u_{{\theta},{K_{\theta}}}\in\mathbb{C}^{M}$, with
integer $K_{\theta}\geq 1$ independent of $\theta$ and $M$, such that
\[
\norm{u_{\theta,{K_\theta}} -u_\theta}_{2}\lesssim\frac{1}{\sqrt{K_{\theta}}%
}.
\]
Similarly, for any value of $\tau$, there exists a $(2K_{\tau}+1)$-sparse
vector $u_{{\tau},{K_{\tau}}}\in\mathbb{C}^{D}$, with integer $K_{\tau}\geq 1$ independent
of $\tau$ and $D$, with
\[
\norm{u_{\tau,{K_\tau}} -u_\tau}_{2}\lesssim\frac{1}{\sqrt{K_{\tau}}}.
\]

\end{lemma}

\begin{proof}
The proof is omitted, and deferred to \cite{Wunder2018_TWC}.
\end{proof}

The error estimates derived in Lemma \ref{lem:SparseAppr} are independent of
the design parameters $M$ and $N$, which means that increasing them results in
an increased \emph{relative sparsity} (number of non-zero to total elements)
for the sparse approximations of $u_{\theta}$ and $u_{\tau}$, respectively.
However, increasing $N$ requires a proportional bandwdith increase, which is
not always available. On the other hand, there are no such limitation on
increasing $M$, which can safely be assumed arbitrarily large.

We now use Lemma \ref{lem:SparseAppr} to quantify the error of approximating
$W(t)$ as having a hierarchical sparse structure. Before proceeding, we will
assume the following two properties for the delay/angle pairs of the channel
paths. The first property imposes a separation among distinct path angles.

\begin{condition}
[Angular separation]For any two distinct path angles $\theta_{i}, \theta_{j}$
($\theta_{i} \neq\theta_{j}$) in the set of delay/angle pairs, it holds
\begin{align*}
\abs{\theta_{i} -\theta_{j}} \geq2K_{\theta}/M,
\end{align*}
whereby $\abs{\cdot}$ is meant in a wrap-around sense over the interval
$[0,2\pi)$.
\end{condition}

Note that angular separations is a standard assumption in this area of
research, see for instance
\cite{candes2014towards,tan2014direction,heckel2014SuperResRadar} and can be
safely assumed to hold for sufficiently large $M$.

The second assumption essentially excludes the possibility of a channel with
excessively many paths having the same angle.

\begin{condition}
[Limited delays-per-angle]For each distinct angle $\theta_{\imath}$ in the set
of paths angle/delay values , there exists at most $K$ delays $\tau_{k}$ such
that $(\tau_{k}, \theta_{\imath})$ is in the set of delay/angle pairs.
\end{condition}

The delays-per-angle condition is furthermore very reasonable on physical
grounds: multipath components with different delays travel over different
paths, hence the probability of arriving at the ULA with the same angle is
very small. By the same argument one excepts that a choice $K=1$ is reasonable
for typical propagation conditions.

We can now characterize the error of approximating $W(t)$ (for any $t$) by a
matrix with a hierarchically sparse structure.

\begin{proposition}
\label{prop:sparsity} Consider an arbitrary slot $t\in\lbrack T]$. Under the
angular separation and delays-per-angle conditions, there exists a matrix
$W_{\Omega}(t)\in\mathbb{C}^{D\times M}$ whose vectorized version is
$(L(2K_{\theta}+1),K(2K_{\tau}+1))$-sparse and satisfies
\[
\norm{ W(t)- W_\Omega(t)}\leqsim\left(  K_{\theta}^{-1}+K_{\tau}^{-1}\right)
\sum_{p=0}^{L}\rho_{p}(t_{i})
\]

\end{proposition}

\begin{proof}
Please see Appendix A.
\end{proof}

Clearly, $W(t),t\in[T]$, can be well approximated by a hierarchical sparse
matrix $W_\Omega(t)$ by choosing $K_{\theta}$ and $K_{\tau}$ sufficiently large. This, in
turn, suggests incorporation of the HiHTP algorithm, treating $W(t),t\in[T]$ as $(L(2K_{\theta}+1),K(2K_{\tau}+1))$-sparse.

\subsection{Error/overhead tradeoff of HiHTP algorithm}

We now want to identify the tradeoff between estimation accuracy and number of subcarriers
and antennas that should be considered in order to achieve it. The following result provides a rigorous characterization of the HiHTP
recovery guarantees, taking into account the probabilistic nature of the
sampling matrices and the error of approximating $\bar{W}(t)$ as hierarchically sparse.

\begin{theorem}
\label{thm: error bound off-grid case} Let the number of sampled antennas and subcarriers satisfy
\begin{align*}
O_{\theta}  &  \geqsim \delta^{-2} L(2K_{\theta}+1) \log^{4}(M),\\
O_{\tau}  &  \geqsim \delta^{-2} K(2K_{\tau}+1) \log^{4}(N),
\end{align*}
respectively, for some $\delta<1/\sqrt{3}$. Then, with a probability larger than $1-M^{-\log^{3}(M)}
-N^{-\log^{3}(N)}$, the sequence of estimates $\hat{\bar{W}}^{(i)}$ generated
by the HiHTP algorithm (Algorithm \ref{alg:HiHTP}) treating $\text{vec}(\bar{W})$ as   $(L(2K_{\theta}+1),K(2K_{\tau}+1),T)$-sparse, will obey the error bound
\begin{align*}
\norm{\bar{W}- \hat{\bar{W}}^{(i+1)}} \leq &  \kappa^{i}
\norm{\bar{W} - \hat{\bar{W}}^{(i)}} + \tau\norm{Z}\\
\  &  + C (K_{\theta}^{-1}+K_{\tau}^{-1}) \sum_{i=1}^{T} \sum_{p=0}^{L} \rho
_{p}(t_{i})
\end{align*}
where $\rho$, $\tau$ are as in Theorem \ref{thm:HiHTP}, and $C$ is a universal constant.
\end{theorem}

\begin{proof}
Please see Appendix B.
\end{proof}

\section{Conclusion}
In this paper, we explored hierarchical sparse estimation framework for massive MIMO
channel estimation. The framework can be used to design appropriate algorithms exploiting
the sparse nature of massive MIMO channel in joint angular and delay domain. In numerical simulations
we show the benefit in terms of reduced pilot overhead, particular for small overhead numbers. We
also extend our analysis to the general 'off-grid' case and derive the sufficient pilot overhead scaling for
perfect recovery for large signal dimension.

\section*{Acknowledgment}

AF and GK acknowledge support from the DFG (Grant KU
1446/18-1), IR and JE from the DFG (EI 519/9-1), the Templeton
Foundation and the ERC (TAQ), MB from the DFG (CA 1340/1-1 and WU 598/7-1,
SH and GC from the DFG (CA 1340/1-1), and GW from the DFG (WU 598/7-1 and WU 598/8-1).
All DFG projects are within the German priority program on 'Compressed Sensing
in Information Processing' (COSIP).

GW acknowledges also support from H2020 project ONE5G (ICT-760809)
receiving funds from the European Union.
The authors would like to acknowledge the contributions of their colleagues
in the project, although the views expressed in this contribution
are those of the author and do not necessarily represent the project.

\appendices

\section{Proof of Proposition \ref{prop:sparsity}}

Let us drop the index $t_{i}$. For each delay/angle pair $(\tau_{p},
\theta_{p})$, we define $\imath_{p}$ and $\jmath_{p}$ through $\imath_{p} :=
\argmin_{i} \abs{ i/{QN} - \tau_p}, \jmath_{p} := \argmin_{j}
\abs{j/M - \theta_p}$, and a rectangle $\Omega_{p} := [\imath_{p} - K_{\theta
}, \imath_{p} +K_{\theta}] \times[\jmath_{p} - K_{\tau}, \jmath_{p} +K_{\tau
}]:= I_{p} \times J_{p} \sse [M] \times[QN]$. Finally define $\Omega$ through
\begin{align*}
\Omega= \bigcup_{p=0}^{L} \Omega_{p}.
\end{align*}
Due to the angular separation and delay-per angle condition, $\Omega$ is an
$(L(2K_{\theta}+1), K(2K_{\tau}+1))$-sparse support (see Figure
\ref{fig:SparsityPattern}.)

\begin{figure}[ptb]
\centering\includegraphics[width=.6\columnwidth]{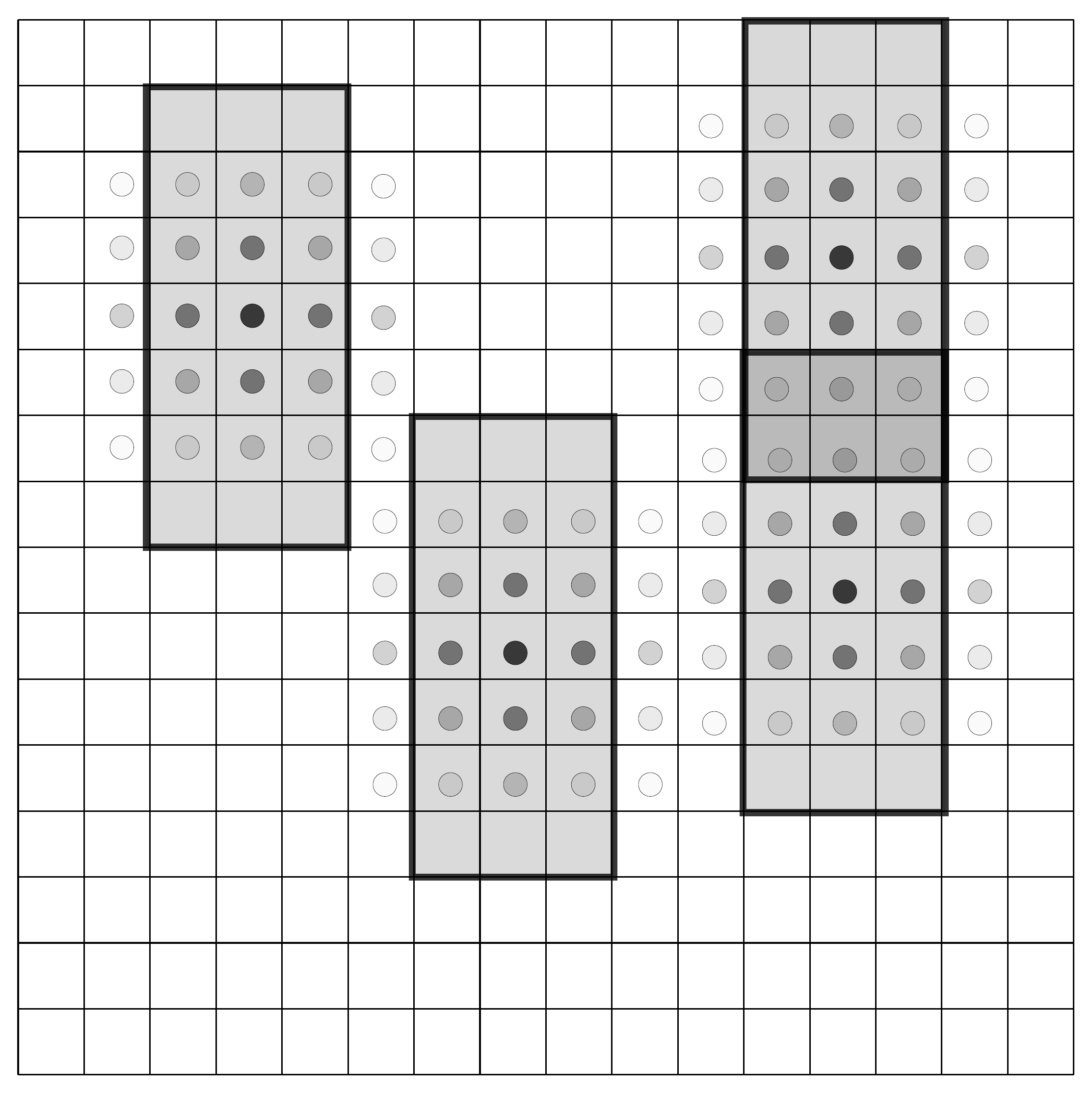} \caption{The
angular separation condition prevents the rectangles $\Omega_{p}$ to intersect
when they are associated to different values of $\theta$. It does not prevent
intersection for different values of $\tau$, but this does not prevent
$K(K_{\tau}+1)$-sparsity in each $\theta$-block. }%
\label{fig:SparsityPattern}%
\end{figure}

Now we estimate
\[
\norm{ W(t_i)- W_\Omega(t_i)}\leq\sum_{p=0}^{L}\rho_{p}(t_{i}%
)\norm{u_{\tau_p} u_{\theta_p}^H- \left(u_{\tau_p}u_{\theta_p}^H\right)_{\Omega_p}}.
\]
For each $p$, we estimate
\begin{align*}
&
\norm{u_{\tau_p} u_{\theta_p}^H- \left(u_{\tau_p}(u_{\theta_p})^H\right)_{\Omega_p}}\\
&  \quad\leq
\norm{u_{\tau_p} u_{\theta_p}^H- \left(u_{\tau_p}\right)_{J_p}(u_{\theta_p})^H}\\
&  \quad\quad
+\norm{(u_{\tau_p})_{J_p} (u_{\theta_p})^H-(u_{\tau_p})_{I_p}(u_{\theta_p})_{I_p})^H}\\
&  \quad
=\norm{u_{\tau_p}-\left(u_{\tau_p}\right)_{J_p}}\norm{u_{\theta_p}^H}\\
&  \quad\quad
+\norm{(u_{\tau_p})_{J_p}}\norm{ (u_{\theta_p})^H-(u_{\theta_p})_{I_p})^H}\\
&  \quad\leqsim K_{\theta}^{-1}+K_{\tau}^{-1}.
\end{align*}
We used that $\norm{u_{\tau}}=\norm{u_{\theta}}=1$ for each $\tau$ and
$\theta$, together with Lemma \ref{lem:SparseAppr}. The claim follows.

\section{Proof of Theorem \ref{thm: error bound off-grid case}}

Let $\Omega$ be the $(L(2K_{\theta}+1), K(2K_{\tau}+1))$-sparse support set
defined in Proposition \ref{prop:sparsity}. Define the set $\widehat{\Omega}
=[0,T-1] \times\Omega$. $\widehat{\Omega}$ is then $(T, L(2K_{\theta}+1),
K(2K_{\tau}+1))$-sparse, and further
\begin{align}
\norm{\bar{W} - W_{\widehat{\Omega}}} \leqsim (K_{\theta}^{-1}+K_{\tau}^{-1})%
\sum_{i=1}^{T} \sum_{p=0}^{L} \rho_{p}(t_{i}). \label{eq:errorBound}%
\end{align}

Now, we utilize that the matrices $A_{\theta}$ and $A_{\tau}$ are quadratic 
DFT-matrices (set $D=N$). Applying Theorem 12.31 from \cite[p.405]%
{MathIntroToCS} together with the assumptions on $O_{\theta}$ and $O_{\tau}$
yields that

\begin{itemize}
\item with probability larger than $1-M^{-\log^{3}(M)}$
\begin{align*}
\delta_{3 L(2K_{\theta}+1)} \left(  \frac{1}{\sqrt{O_{\theta}}}P_{\theta
}A_{\theta}^{*}\right)  \leq\widetilde{\delta},
\end{align*}

\item with probability larger than $1-N^{-\log^{3}(N)}$
\[
\delta_{3K(2K_{\tau}+1)}\left(  \frac{1}{\sqrt{O_{\tau}}}P_{\tau}A_{\tau
}\right)  \leq\widetilde{\delta},
\]

\end{itemize}

where $\widetilde{\delta}$ is defined through $\delta= \widetilde{\delta
}(2+\widetilde{\delta})$, which in particular implies that $\widetilde{\delta}
\leq\delta/2\leq1/(2\sqrt{3})$, so that
\begin{align*}
\delta\geqsim \widetilde{\delta}.
\end{align*}
Hence, an estimate of the form $\geqsim \delta^{-2}$ implies an estimate of
the form $\geqsim \widetilde{\delta}^{-2}$ (with another constant), whence the
theorem is applicable.

Since further $I_{T}$ trivially obeys $\delta_{T}=0$, Theorem \ref{thm:HiRIP}
implies that
\begin{align*}
\delta_{(3 \cdot L(2K_{\theta}+1),3 \cdot(L_{2}(2K_{\tau}+1), T)}(\Psi)  &
\leq(1+\delta)^{2} -1 = \delta(\delta+2)\\
&  \leq\frac{1}{\sqrt{3}}.
\end{align*}
Theorem \ref{thm:HiHTP} implies that
\begin{align*}
\Vert W_{\widehat{\Omega}}-W_{\widehat{\Omega}}^{\left(  t+1\right)  }%
\Vert\leq\rho^{t}\Vert W_{\widehat{\Omega}}-W_{\widehat{\Omega}}^{\left(
t\right)  }\Vert+\tau\Vert Z\Vert.
\end{align*}
The claim now follows from
\begin{align*}
\norm{W-W^{(r)}} \leq\Vert W_{\widehat{\Omega}}-W^{\left(  r\right)  } \Vert+
\norm{W - W_{\widehat{\Omega}}}
\end{align*}
for $r=0$ and $r=(t+1)$ and the error bound \eqref{eq:errorBound}.

\bibliographystyle{IEEEtran}
\bibliography{mMIMOce}

\end{document}